\documentclass{article}
\usepackage[utf8]{inputenc}
\usepackage{mathrsfs}
\usepackage{amsmath}
\usepackage{amssymb}
\usepackage{amsthm}
\usepackage{authblk}
\usepackage{tikz}
\usetikzlibrary{automata, positioning}
\usepackage{float}
\usepackage{tabularx,lipsum,environ}
\usepackage[ruled,vlined]{algorithm2e}

\usepackage{hyperref}
\hypersetup{
    colorlinks=true,
    linkcolor=blue,
    filecolor=magenta,      
    urlcolor=cyan,
}

\newtheorem{observation}{Observation}
\newtheorem{theorem}{Theorem}
\newtheorem{lemma}{Lemma}

\providecommand{\keywords}[1]
{
  \small	
  \textbf{\textit{Keywords---}} #1
}

\title{Approximation algorithm for finding short synchronizing words in weighted automata}
\author{Jakub Ruszil}
\affil{Jagiellonian University}
\date{March 2021}

\begin{document}

\maketitle

\begin{abstract}
In this paper we are dealing with the issue of finding possibly short synchronizing words in automata with weight assigned to each letter in the alphabet $\Sigma$. First we discuss some complexity problems, and then we present new approximation algorithm in four variations.
\end{abstract}
\keywords{formal languages, automata theory, \v{C}ern\'{y} conjecture, synchronization, algorithms}

\section{Preliminaries}
Finite automata are widely used to model variety of systems working in discrete mode. One of active areas of research in a theory of finite automata is a synchronization problem, which can be understood as a problem of resetting a system to a chosen state, from any other state \cite{cerny}. For exhaustive survey of synchronization problem in automata theory see \cite{volkov2}. A generalization of finite automata, where to every letter is attached a weight or a cost can be used to model a situation, in which different actions on system takes different amount of time. Such a modification of finite automata where studied for by number of authors, also in an aspect of synchronization, see for example \cite{volkov}, \cite{doyen}, \cite{kretinsky},\cite{szabolcs}, \cite{mohri}.\newline  
A \emph{weighted deterministic finite automaton} (WFA) $\mathcal{A}$ is an ordered tuple $(Q, \Sigma, \delta, \gamma)$ where $Q$ is a finite set of states, $\Sigma$ is a finite alphabet, $\delta:{Q \times \Sigma}\rightarrow{Q}$ is a transition function and $\gamma: \Sigma \rightarrow{\mathbb{N}}$ is a weight function.  For $\emph{w} \in \Sigma^\ast$ and $\emph{q} \in Q$ we define $\delta(\emph{q},\emph{w})$ inductively as $\delta(\emph{q},\epsilon) = q$ and $\delta(\emph{q},\emph{aw}) = \delta(\delta(\emph{q},\emph{a}), \emph{w})$ for $a \in \Sigma$ where $\epsilon$ is the empty word. We write $q.w$ instead of $\delta(q,w)$ wherever it does not cause ambiguity. We define the  weight $\gamma(w)$ of a word $w \in \Sigma^\ast$ as a sum of weights of all letters in $w$. By $|w|$ we denote the number of letters in $w$. For a given WFA $\mathcal{A}=(Q,\Sigma,\delta,\gamma)$ a word $\emph{w} \in \Sigma^\ast$ is called \emph{synchronizing} if there exists $\overline{q} \in Q$ such that for every $\emph{q} \in Q$, $q.w = \overline{q}$. A WFA is called \emph{synchronizing} if it admits any synchronizing word. Let $\mathcal{L}_n = \{\mathcal{A} = (\Sigma, Q, \delta, \gamma): \mathcal{A}\;is\;synchronizing\;and\;|Q| = n\}$. We define $d(\mathcal{A}) = min\{\gamma(w):w\; is\; synchronizing\:word\; for\; \mathcal{A}\}$ and $d(n) = max\{d(\mathcal{A}) : \mathcal{A} \in \mathcal{L}_n\}$. All above definitions and notations also apply to \emph{deterministic finite automata} (DFA) with constant cost function. Notice that problem of synchronization of a WFA is a generalization of synchronization for DFA. It becomes the problem of synchronization of DFA, if we define weights for all letters equal. Before moving further we make observation that follow from the classic results concerning theory of automata synchronization.
\begin{observation}
\label{observation:pair_aut}
Let $\mathcal{A}$ be a WFA. Then $\mathcal{A}$ is synchronizing if and only if, for each $p, q \in Q$ there exist a word $w$, such that $p.w = q.w$.
\end{observation}
Notice that problem of synchronization of a WFA is a generalization of synchronization for DFA. It becomes the problem of synchronization of DFA, if we define weights for all letters equal. It is not hard to find an example of a WFA, for which the shortest synchronizing word (with minimal number of letters) is not equal to the one with minimal weight. Example of such automaton is depicted on the figure below. 
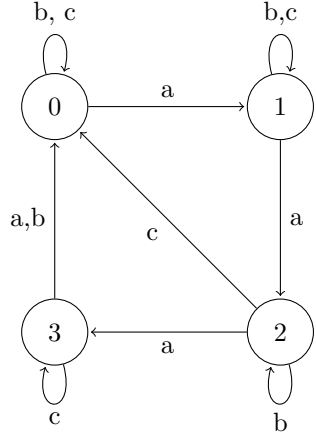
\begin{figure}[H]
    \centering
    \begin{tikzpicture}[shorten >=1pt,node distance=3cm,on     grid,auto] 
        \node[state] (0)   {$0$}; 
        \node[state] (1) [right=of 0] {$1$}; 
        \node[state] (2) [below=of 1] {$2$}; 
        \node[state] (3) [left=of 2] {$3$};
        \path[->] 
        (0) edge node {a} (1)
                    edge [loop above] node {b, c} ()
        (1) edge  node {a} (2)
            edge [loop above] node {b,c } ()
        (2) edge  node {a} (3)
        edge  node {c} (0)
        edge [loop below] node {b} ()
        (3) edge  node {a,b} (0)        edge [loop below] node {c} ();

        \end{tikzpicture}
          \label{fig:example}
    \caption{Automaton $\mathcal{A}$}
\end{figure}
Let us define $\gamma(a) = 1$, $\gamma(b) = 1 $, $\gamma(c) = 6$ - it is easy to check, that although the shortest synchronizing word for $\mathcal{A}$ is $baacb$, the word with minimal weight is $baaabaaab$.\newline
With all those definitions and observations we may now introduce a new, polynomial-time algorithm for approximating shortest synchronizing word in WFA.

\section{Algorithm}

We consider the following modification of our definition of WFA. Let $\mathcal{A} = (Q, \Sigma, \delta, \phi)$ be a (\textit{extended WFA}) eWFA, where $Q$, $\Sigma$ and $\delta$ are as in WFA, but $\phi: \Sigma \times Q \rightarrow{\mathbb{N}}$ is a weight function which assigns weights to transitions rather than to letters. Let $\phi: Q \times \Sigma^* \rightarrow{\mathbb{N}}$ be a weight function. Let also $w = a_1...a_k$ and $a_i \in \Sigma$. We define $\phi(q,w) = \sum_{i=0}^{k-1} \phi(q.(a_1...a_i), a_{i+1})$. We also define the weight of synchronization of $\mathcal{A}$ with a word $w \in \Sigma^*$ as $\phi(w) = \max_{q \in Q} \phi(q, w)$.  Define the following decision problem:\newline
\emph{E-WEIGHTED-SYNCHRO-WORD: Given an eWFA $\mathcal{A}$ and a positive integer l, is it true that $\mathcal{A}$ admits a synchronizing word $w$ such that $\phi(w) \leq l$?}
\newline In 2012 Volkov et. al. \cite{volkov} proved that \emph{E-WEIGHTED-SYNCHRO-WORD} is \emph{NP}-complete if $l$ is given in unary representation, however is \emph{PSPACE}-complete, when $l$ is given in binary representation, because the shortest synchronizing words can be exponentially long in $|Q|$. We are going to discuss a similar problem but for WFA. Next, we give an algorithm which, for a given WFA, returns a weighted synchronizing word of weight $O(k\cdot |Q|^3)$, where $k$ is the maximum weight of a (single) letter, if such a word exists. 
Before moving forward we first discuss complexity issues of a problem of finding shortest synchronizing words in a WFA. Consider the following problem:\newline
\emph{WEIGHTED-SYNCHRO-WORD: Given a WFA $\mathcal{A}$ and a positive integer l, is it true that $\mathcal{A}$ admits a synchronizing word $w$ such that $\gamma(w) \leq l$?}\newline
\begin{lemma}
\label{lemma:np}
\emph{WEIGHTED-SYNCHRO-WORD} $\in$ NP.
\end{lemma}
\begin{proof}
To prove that \emph{WEIGHTED-SYNCHRO-WORD} is in NP it suffices to show that there exists a polynomial algorithm that verifies if an instance $( \mathcal{A},l,w)$ of the problem with $w \in \Sigma^*$, belongs to \emph{WEIGHTED-SYNCHRO-WORD}.
Namely we need to check if a given word $w$, with $\gamma(w) \leq l$ synchronizes $\mathcal{A}$. To this end we check, if there exists a $q \in Q$, such that for all $p \in Q$ we have that $p.w = q$. It is obvious that we can do it in time $O(|Q|\cdot |w|)$ and we claim the lemma holds.
\end{proof}
Notice that here we assumed, that $l$ was given in unary representation. Corollary of our further considerations is a fact, that if WFA is synchronizing, then it admits synchronizing word of weight at most polynomially long in a number of states and that proves, that \emph{WEIGHTED-SYNCHRO-WORD} is in NP even if $l$ is given in binary representation. Recall the following theorem from \cite{gawrychowski}.
\begin{theorem}[Gawrychowski, Straszak]
\label{theorem:gawrych}
For every constant $\epsilon > 0$, unless $P = NP$, no polynomial algorithm approximates shortest synchronizing word for a given DFA within a factor of $n^{1-\epsilon}$.
\end{theorem}
Having that, we can construct a simple reduction from a problem of approximating shortest synchronizing word in DFA, to a problem of finding synchronizing word of minimal weight in WFA. Notice that the following theorem, together with Lemma \ref{lemma:np} also implies $NP$-completeness of \emph{WEIGHTED-SYNCHRO-WORD}. Let us state and prove the following theorem. 
\begin{theorem}
\label{theorem:approx}
For every constant $\epsilon > 0$, unless $P = NP$, no polynomial algorithm approximates weighted synchronizing word for a given WFA with weights bounded by some constant $k$, within a factor of $k\cdot n^{1-\epsilon}$.
\end{theorem}
\begin{proof}
Let $\mathcal{A} = (Q, \Sigma, \delta)$ be a DFA. We construct $\mathcal{B} = (Q, \Sigma, \delta, \gamma)$, such that $\Sigma$, $Q$ and $\delta$ are such as in $\mathcal{A}$, and for any $a \in \Sigma$, we define $\gamma(a) = k$. Let us suppose that there exists an algorithm $f$ that approximates minimal weight of a synchronizing word  for a bounded WFA. That would imply the existence of algorithm $f_2$ that approximates shortest synchronizing word in a given DFA, as follows. Algorithm $f_2$ constructs a WFA $\mathcal{B}$ for a given DFA $\mathcal{A}$ in a way described above (notice that this construction can be done in polynomial time), runs $f$ for $\mathcal{B}$ and returns the word computed by $f$. It can be easily seen that such a word would approximate shortest synchronizing word for $\mathcal{A}$ within a factor of $n^{1-\epsilon}$ and that is a contradiction.
\end{proof}
Before proceeding, we introduce the following definition. With any WFA $\mathcal{A}$ we associate another WFA, $\mathcal{A}^m = (\binom{Q}{\leq m}, \Sigma, \tau, \gamma)$, where $\binom{Q}{\leq m}$ stands for the set of subsets of $Q$ of cardinality at most $m$, $\Sigma$ and $\gamma$ are as in $\mathcal{A}$ and the transition function $\tau:{\binom{Q}{\leq m} \times \Sigma}\rightarrow{\binom{Q}{\leq m}}$ is defined for $Q' \in \binom{Q}{\leq m}$ and $a \in \Sigma$ as follows: $\tau(Q', a) = \bigcup_{q \in Q'}\delta(q, a)$. Notice that for $m = |Q|$ we obtain classic \emph{power automaton} construction. \newline
Now we are ready to introduce our algorithm. In the following we fix $m \in \{2, ..., |Q|\}$, define $t:Q\times Q\rightarrow{\{0,1\}}$ and $h:Q\times Q\times \Sigma^*\rightarrow{\mathbb{R}_+}$. More detailed description of $t$ and $h$ is discussed further on. First we describe two procedures \textbf{$m$-COMPUTE-WORDS} and \textbf{$m$-APPROXIMATE-WEIGHT-SYNCH} and then we prove their correctness and estimate their time complexity.\newline
\begin{algorithm}[H]
\SetKwInOut{Input}{input}
\SetKwInOut{Output}{output}
\LinesNumbered 
\Input{WFA $\mathcal{A} = (Q, \Sigma, \delta, \gamma)$}
\Output{Set $W = \{ (Q', w) : Q' \in \binom{Q}{\leq m}, w \in \Sigma^\ast \}$}
 $W = \varnothing$\;
 Compute $\mathcal{A}^m$ for $\mathcal{A}$\;
 $S = \varnothing$\;
 \For{singleton $p$ in $\mathcal{A}^m$ }{
    compute shortest paths to $p$ from all other possible states from $\mathcal{A}^m$ and add them to $S$\;
 }
    \For{$s$ in $S$}{
        produce word $w$ from path $s$\;
        $Q = $ first state on $s$\;
        \If{there is $(P, u) \in W$ such that $P = Q$ and $\gamma(u) > \gamma(w)$}{
         remove $(P,u)$ from $W$\;
        }
    \If{there is no $(P, u) \in W$ such that $P = Q$}{
    add $(Q,w)$ to $W$\;
    }
    }
 return $W$\;
 \caption{$m$-COMPUTE-WORDS}
 \label{algorithm:comp_words}
\end{algorithm}

\begin{algorithm}[H]
\SetKwInOut{Input}{input}
\SetKwInOut{Output}{output}
\LinesNumbered 
\Input{WFA $\mathcal{A} = (Q, \Sigma, \delta, c)$}
\Output{ $w \in \Sigma^\ast$ such that $w$ synchronizes $\mathcal{A}$ or $\epsilon$}
 $W =$ $m$-COMPUTE-WORDS($\mathcal{A}$, $m$)\;
 $T = Q$\;
 $u = \epsilon$\;
\While{$|T| > 1$}{
    choose $(P,w)$ from $W$ such that $t(P, T)$ is true and $h(P,T,w)$ is minimal\;
    \If{such $(P,w)$ does not exist}{
    return $\epsilon$\;
    }
    $u = uw$\;
    $T = T.w$
}
return $u$\;
 \caption{$m$-APPROXIMATE-WEIGHT-SYNCH}
 \label{algorithm:approx}
\end{algorithm}

We can define $t(P,T)$ and $h(P,T,w)$ in a few ways. First let $t(P,T)$ be a predicate $t(P, T) \equiv (P \subseteq T \wedge |P| > 1)$ and $h(P,T,w) = \frac{\gamma(w)}{|P| - |P.w|}$. Denote that choice of $t$ and $h$ as H1. Definition of $h$ and $t$ is motivated by the will to choose not necessarily the word with minimal possible weight in each iteration, but to balance between the weight and the ability of a word to reduce the cardinality of the set $T$ in each step. We propose also another three heuristics for choosing words $(P,w)$ in line 5 of Algorithm \ref{algorithm:approx}. Define $t(P,T)$ in the same manner as in H1 and $h(P,T,w) = \frac{\gamma(w)}{|T| - |T.w|}$ and denote this choice as H2. The reason for this definition is a possibility of finding words in $W$, that may not be the shortest ones, but reduce more states in $T$ and therefore are a good choice. Next heuristic we propose is as follows: define $t(P,T) \equiv (P \subseteq T \wedge |P| = \min(m,|T|))$ and  $h(P,T,w) = \gamma(w)$ and denote this choice as H3. Here we model a greedy approach, known from the classic Eppstein's heuristic algorithm for finding synchronizing words for a given DFA \cite{eppstein}. The last variant, denoted as H4, is a modification of H2 above where we let  $t(P,T)$ be as in H2 and H1 and $h(P,T,w) = \frac{\gamma(w)}{(|T| - |T.w|)^2}$. Here we suppose that the ability of a word to send many states into one is more important than the weight $\gamma(w)$. It is likely, that using H1, H2 or H3, the Algorithm \ref{algorithm:approx} will take $a \in \Sigma$ in first iteration, which can be not an optimal choice in numerous situations, for example \v{C}ern\'{y} automata \cite{cerny}. Using $H4$ we want to prevent this situation, by taking square in a denominator. For the sake of next proofs we define $k = max\{\gamma(a)\colon a \in \Sigma\}$.

\begin{lemma}
\label{lemma:min}
For each element $(P,w) \in W$ returned by $m$-COMPUTE-WORDS procedure we have $|P.w| = 1$ and $w$ is a word with minimal weight having that property.
\end{lemma}
\begin{proof}
Consider procedure $m$-COMPUTE-WORDS. In a \textbf{for} loop in line 4 we compute paths to all singletons in $\mathcal{A}^m$ from all possible other states using the Dijkstra algorithm. After that loop every element in $S$ consists of a sequence of vertices from some state to a singleton. It is obvious that any such path corresponds to a word $w$, such that $|P.w| = 1$, where $P$ is the first element of that path. Minimality of $\gamma(w)$ is immediate from the fact that $S$ contains paths with minimal weight and that concludes the proof.
\end{proof}
\begin{lemma}
\label{lemma:estim}
If $\mathcal{A}$ is synchronizing then for each $(P,w) \in W$ returned by $m$-COMPUTE-WORDS procedure we have $c(w) \leq (|P|-1)\cdot k\cdot \binom{|Q|}{2}$.
\end{lemma}
\begin{proof}
From Observation \ref{observation:pair_aut} we have that for any $p, q \in Q$ there exists a word $u$, such that $p.u = q.u$. Take the shortest $u$ with such property. Since $\mathcal{A}$ is deterministic and $u$ is minimal, for each prefixes $v_1 \neq v_2$ of $u$, such that $v_1$ is a prefix of $v_2$ it holds, that $\{p,q\}.v_1 \neq \{p,q\}.v_2$ and $|\{p,q\}.v_1| \geq |\{p,q\}.v_2|$. That implies $\gamma(u) \leq k \cdot \binom{|Q|}{2}$, since the weight of every letter is bounded by $k$. Hence, we can construct a word $u_1u_2\ldots u_{|P|-1}$, such that $|P.u_1u_2\ldots u_{|P|-1}| = 1$. On the other hand from Lemma 1 we have that $|P.w| = 1$ and $\gamma(w)$ is minimal among all words with that property, so we have $\gamma(w) \leq \gamma(u_1u_2\ldots u_{|P|-1}) \leq (|P|-1)\cdot k\cdot \binom{|Q|}{2}$. 
\end{proof}

\begin{lemma}
\label{lemma:disamb}
If $\mathcal{A}$ is synchronizing then for each $P \subset Q$, such that $|P| \leq m$, there is exactly one element of $W$ having $P$ as the first coordinate.
\end{lemma}
\begin{proof}
Using similar argument as in the proof of Lemma \ref{lemma:estim}, we can easily prove that there exists at least one such element in W. To prove that there is at most one such element, consider the \textbf{if} statement in line $9$. It is obvious that after execution of that \textbf{if}, there can be at most one element with $P$ as the first coordinate in $W$ and that ends the proof.
\end{proof}
\begin{theorem}
\label{theorem:main1}
Procedure $m$-APPROXIMATE-WEIGHT-SYNCH using H1,H2,H3 or H4 returns a weighted synchronizing word $u$, such that $\gamma(u) \leq k \cdot (|Q|-1) \cdot \binom{|Q|}{2}$ if such a word exists.
\end{theorem}
\begin{proof}
Consider the \textbf{while} loop in $m$-APPROXIMATE-WEIGHT-SYNCH procedure. At the start of that loop we have $T = Q$. In line 5 the algorithm chooses an element of $W$, that "reduces" the cardinality of $T$. Notice that from Observation \ref{observation:pair_aut} such an element exists if $\mathcal{A}$ is synchronizing and $|T| > 1$. If at any step of the \textbf{while} loop the algorithm cannot find an element that fulfills $t(P,T)$, we also can deduce from Observation \ref{observation:pair_aut} that there exists at least one pair $p,q \in Q$ that cannot be "reduced" and by that there is no synchronizing word for that automaton. Notice that $|W|$ is $O(|Q|^m)$. Also, any $w$ associated to $P$ in $W$ is of length $O(|Q|^m)$, because any set on path from $P$ to a singleton can be visited at most once, and, since $\mathcal{A}$ is deterministic, we know that for each prefix $w_1$ of $w$ we have that $|P.w_1| \leq |P|$. So we know that $m$-APPROXIMATE-WEIGHT-SYNCH in polynomial time decides if $\mathcal{A}$ is synchronizing. Now we estimate the weight of $u$ returned by $m$-APPROXIMATE-WEIGHT-SYNCH. Assume that $\mathcal{A}$ is synchronizing. Denote by $(P_i,w_i)$ the pair chosen in $i$-th step of \textbf{while} loop. Also denote $T_1 = Q$, and $T_i = T_{i-1}.w_{i-1}$ for $i \in \{1, .., |Q|-1\}$. If for any $j < |Q| - 1$ we have $|T_{j}.w_j| = 1$, we define $w_l = \epsilon$ and $P_l = T_j.w_j$, for all $l > j$.  It is obvious that $P_i \subseteq T_i$. On the other hand, from Lemma \ref{lemma:estim} for each $i$ we have that $\gamma(w_i) \leq (|P_i| - 1)\cdot k \cdot \binom{|Q|}{2}$. So we have $\gamma(u) = \gamma(w_1w_2..w_{|Q|-1}) \leq \sum_{i=1}^{|Q|-1} (|P_i|-1)\cdot k \cdot \binom{|Q|}{2} = k \cdot \binom{|Q|}{2} \cdot (\sum_{i=1}^{|Q|-1} (|P_i|-1))$. Notice that for each $i$, since $P_i \subseteq T_i$ it holds, that $|T_{i}| - |T_{i + 1}| \geq |P_{i}| - 1$. Also $\sum_{i=1}^{|Q|-1} (|T_{i}| - |T_{i + 1}|) = |Q| - 1$. That implies $\sum_{i=1}^{|Q|-1} (|P_i|-1)) \leq |Q| - 1$ and that concludes the proof.
\end{proof}
\begin{theorem}
\label{theorem:complexity}
Procedure $m$-APPROXIMATE-WEIGHT-SYNC using H1, H2, H3 or H4 works in time $O(|\Sigma|\cdot |Q|^{m+1} \log |Q| + |Q|^{m+3})$.
\end{theorem}
\begin{proof}
We first consider the procedure $m$-COMPUTE-WORDS. Computing $\mathcal{A}^m$ can be done in time $O(|\Sigma|\cdot |Q|^{m+1})$. Notice that there are exactly $|Q|$ singletons in $\mathcal{A}^m$, so the loop in line $4$ is executed exactly $|Q|$ times. It is easy to see that the set of vertices of the underlying digraph of that automaton is of size $O(|Q|^{m})$ and the edge set size is $O(|\Sigma|\cdot|Q|^{m})$. 
If we use Dijkstra algorithm to compute shortest paths, we obtain time complexity of the for loop in line $4$ is $O(|Q| \cdot |\Sigma|\cdot|Q|^{m} \log |Q|^{m}) = O(|\Sigma|\cdot |Q|^{m+1}\cdot m\log |Q|) = O(|\Sigma|\cdot |Q|^{m+1}\log |Q|)$.
Now consider the \textbf{for} loop in line 6. We can implement Dijkstra algorithm to obtain words instead of paths, so line 7 does not affect the time complexity. Lines 9--12 can be executed in expected constant time using a dictionary, so the time complexity of the loop is $O(|S|) = O(|Q|^{m+1})$. Joining it all together we obtain that the time complexity of $m$-COMPUTE-WORDS procedure is $O(|\Sigma|\cdot |Q|^{m+1} \log |Q|)$. Now we estimate time complexity of $m$-APPROXIMATE-WEIGHT-SYNC. COMPUTE-WORDS is executed once, and lines 2 and 3 have constant complexity. Consider the \textbf{while} loop. Notice, that at the end of the \textbf{while} block we have either $|T.w| < |T|$ or we return $\epsilon$, so this loop is executed at most $|Q| - 1$ times. From earlier considerations and from Lemma \ref{lemma:disamb} we obtain that the size of $W$ is $O(|Q|^m)$, so choosing proper $(P,w)$ takes at most $O(|Q|^m)$ time. Lines 6, 7 and 8 are executed in constant time. Since $|T| \leq |Q|$ and $|w| \leq (m - 1)\binom{|Q|}{2}$, we can compute $T.w$ in $O(|Q|^3)$. From all that we have that $m$-APPROXIMATE-WEIGHTED-SYNC works in time $O(|\Sigma|\cdot |Q|^{m+1} \log |Q| + |Q|^{m+3})$ and that concludes this theorem.
\end{proof}
Another remark is that even for $m = 2$ we obtain a weighted synchronizing word of weight $O(k \dot |Q|^3)$. If we choose greater $m$, we can shorten the intermediate words $w$ but at a cost of higher time complexity. \newline
We also provide an example of an automaton $\mathcal{B}$ (depicted below) that illustrates the difference between different heuristics. We define $\gamma(a) = 1, \gamma(b) = 6, \gamma(c) = 2, \gamma(d) = 7$. Results of WEIGHTED-SYNCHRO-WORD for different values of $m$ and different heuristics are presented in a Table \ref{table:results}. They were obtained using author's implementation of described algorithm, which can be found under this link \url{https://colab.research.google.com/drive/14MmedWbN83KB-_t-rhaZwpexBaNBfQAa?usp=sharing}. 

\begin{table}[H]
\begin{center}
\begin{tabular}{ |p{1cm}|p{8.5cm}|p{1cm}|p{1cm}|  }
\hline
 Heur. & Word & Length & Weight\\
\hline
\multicolumn{4}{|c|}{$m = 2$} \\
\hline
H1& $caacaacaacaacaacbcaaaabcaaaabcaabcadadadadbcadad$ &48 &120 \\
H2& $caacaacaacaacaacbcaaaabcaaaabcaabcadadadadbcadad$ &48 &120\\
H3& $caacaacaacaacaacbcaaaabcaaaabcaabcadadadadbcadad$ &48 &120\\
H4& $dadbcadadaaabcadadaaaaaaaadbcadad$ & 33&105\\
\hline
\multicolumn{4}{|c|}{$m = 3$} \\
\hline
H1& $caacaacaacaacaacbcaaaabcaaaabcaabcadadadadbcadad$ &48 &120\\
H2& $caacaacaacaacaacbcaaaabcaaaabcaabcadadadadbcadad$ &48 &120\\
H3& $bcaaabcaaabcaaabcdadbcadadadbcadad$ & 34&112\\
H4& $dadadadbcadadabcdadad$ &21 &87\\
\hline
\multicolumn{4}{|c|}{$m = 4$} \\
\hline
H1& $caacaacaacaacaacbcaaaabcaaaabcaabcadadadadbcadad$ &48 &120\\
H2& $caacaacaacaacaacbcaaaabcaaaabcaabcadadadadbcadad$ &48 &120\\
H3& $bcaadaabcaadaabcaadbcaaaaaaaaadadbcadad$ &39 &111\\
H4& $dadadadadbcadadacdadad$ &22 &89\\
\hline
\multicolumn{4}{|c|}{$m = 5$}\\
\hline
H1& $caacaacaacaacaacbcaaaabcaaaabcaabcadadadadbcadad$ &48 &120\\
H2& $caacaacaacaacaacbcaaaabcaaaabcaabcadadadadbcadad$ &48 &120\\
H3& $bcadadaaabcadadadadadbcadad$ & 27 & 79\\
H4& $dadadadadbcadadacdadad$ &22 & 89\\
\hline
\multicolumn{4}{|c|}{$m = 6$} \\
\hline
H1& $caacaacaacaacaacbcaaaabcaaaabcaabcadadadadbcadad$ &48 &120\\
H2& $caacaacaacaacaacbcaaaabcaaaabcaabcadadadadbcadad$ &48 &120\\
H3& $dbcadadaaaadbcadadadadbcadad$ & 28&106\\
H4& $bcaaaaaaadadadadbcadadadbcadad$ & 30&102\\
\hline
\end{tabular}
\end{center}
\label{table:results}
\caption{Results for various heuristics and values of $m$ for $\mathcal{B}$}
\end{table}
The shortest synchronizing word for automaton $\mathcal{B}$ is $dadadadadadabdcadad$ (19 letters, weight 79) and the word with minimal weight is $bcaaaaaadadadadadbcadad$ (23 letters, weight 77). From the results in Table \ref{table:results} we can deduce, that not always choosing higher $m$ does the job. Also, in this particular example heuristics H4 works the best for almost all values of $m$ we chose, except for $m=5$, for which $H3$ is the best choice. 
\begin{figure}[H]
    \centering
   \begin{tikzpicture}
    \foreach \phi in {0,...,11}{
        \node[state,fill=white] (v_\phi) at (360/12 * \phi:4cm) {$\phi$};
      }; 
      \path[->] 
        (v_0) edge node[right] {a,c} (v_1)
        (v_1) edge node[above right] {a} (v_2)
        (v_2) edge node[above] {a} (v_3)
        (v_3) edge node[above] {a} (v_4)
        (v_4) edge node[above left] {a} (v_5)
        (v_5) edge node[left] {a} (v_6)
        (v_6) edge node[left] {a} (v_7)
        (v_7) edge node[below left] {a} (v_8)
        (v_8) edge node[below] {a} (v_9)
        (v_9) edge node[below] {a} (v_10)
        (v_10) edge node[below right] {a} (v_11)
        (v_11) edge node[right] {a,b} (v_0)
        (v_5) edge [bend right = 20] node[right,pos=0.4]  {d} (v_4)
        (v_11) edge [loop right] node {c,d} ()
        (v_0) edge [loop right] node {b,d} ()
        (v_1) edge [loop right] node {b,c,d} ()
        (v_2) edge [loop above] node {b,c,d} ()
        (v_3) edge [loop above] node {b,c,d} ()
        (v_4) edge [loop above] node {b,c} ()
        (v_5) edge [loop left] node {b,c,d} ()
        (v_6) edge [loop left] node {b,c,d} ()
        (v_7) edge [loop left] node {b,c,d} ()
        (v_8) edge [loop below] node {b,c,d} ()
        (v_9) edge [loop below] node {b,c,d} ()
        (v_10) edge [loop below] node {b,c,d} ();
\end{tikzpicture}
          \label{fig:example2}
\caption{Automaton $\mathcal{B}$}
\end{figure}
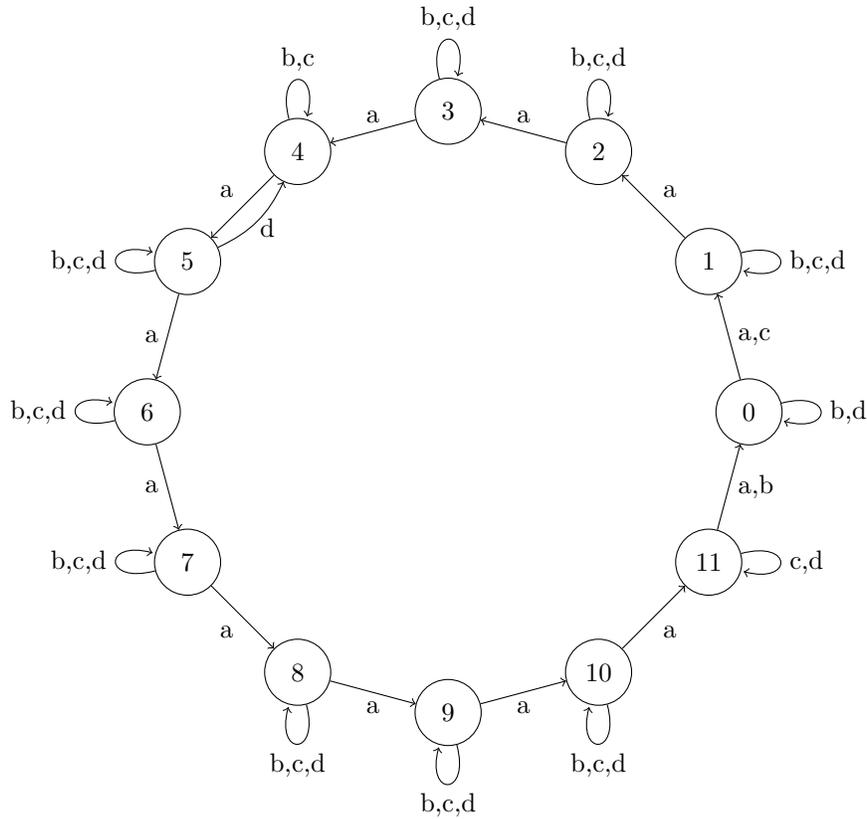

\end{document}